\documentclass[journal]{IEEEtran}
\usepackage{amsmath,amsfonts,amsthm}
\usepackage{algorithm}
\usepackage{array}
\usepackage[caption=false,font=normalsize,labelfont=sf,textfont=sf]{subfig}
\usepackage{textcomp}
\usepackage{stfloats}
\usepackage[hyphens]{url}
\usepackage{verbatim}
\usepackage{graphicx}
\usepackage{cite}

\usepackage{tikz}
\usepackage{enumerate}
\usepackage{xcolor}
\usepackage{orcidlink}
\usepackage{algpseudocode,algorithmicx}
\newcommand*\Let[2]{\State #1 $\gets$ #2}

\DeclareMathOperator{\Rel}{Rel}
\newcommand{\T}{\mathsf{T}}

\newcommand{\tmin}{\theta_{\mathrm{min}}}
\newcommand{\tmax}{\theta_{\mathrm{max}}}
\newcommand{\tdis}{\theta_{\mathrm{dis}}}
\newcommand{\tre}{\theta_{\mathrm{re}}}

\theoremstyle{plain}
\newtheorem{theorem}{Theorem}
\newtheorem{proposition}[theorem]{Proposition}

\newtheorem{problem}{Problem}

\theoremstyle{definition}

\theoremstyle{remark}

\begin{document}

\title{Vertex Addition to a Ball Graph With Application to Reliability\\ and Area Coverage in Autonomous Swarms}

\author{Calum Buchanan\orcidlink{0000-0002-7381-8060}, Puck Rombach\orcidlink{0000-0002-8374-0797}, James Bagrow\orcidlink{0000-0002-4614-0792}, and Hamid R. Ossareh\orcidlink{0000-0002-4964-569X}
\thanks{This work has been submitted to the IEEE for possible publication. Copyright may be transferred without notice, after which this version may no longer be accessible. This work was supported in part by NASA under cooperative agreement VT-80NSSC20M0213. The work of C.~Buchanan was also supported in part by the Vermont Space Grant Consortium Graduate Fellowship Program. A preliminary version of this paper, entitled {\em Node placement to maximize reliability of a communication network with application to satellite swarms}, was presented at the {\em 2023 IEEE International Conference on Systems, Man, and Cybernetics (SMC)}, Honolulu, Oahu, HI, USA. {\em (Corresponding author: Calum Buchanan.)}}
\thanks{Calum Buchanan, Puck Rombach, and James Bagrow are with the Department of Mathematics and Statistics, University of Vermont, Burlington, VT 05405 USA (e-mail: calumjmb@comcast.net; puck.rombach@uvm.edu; james.bagrow@uvm.edu).}
\thanks{Hamid R. Ossareh is with the Department of Electrical and Biomedical Engineering, University of Vermont, Burlington, VT 05405 USA (e-mail: hossareh@uvm.edu).}
\thanks{Implementation in Python 3 of the algorithms described herein is available on-line at \url{https://doi.org/10.5281/zenodo.16746466}.}
}



\maketitle

\begin{abstract}
A unit ball graph consists of a set of vertices, labeled by points in Euclidean space, and edges joining all pairs of points within distance 1. These geometric graphs are used to model a variety of spatial networks, including communication networks between agents in an autonomous swarm. In such an application, vertices and/or edges of the graph may not be perfectly reliable; an agent may experience failure or a communication link rendered inoperable. With the goal of designing robust swarm formations, or unit ball graphs with high reliability (probability of connectedness), in a preliminary conference paper we provided an algorithm with cubic time complexity to determine all possible changes to a unit ball graph by repositioning a single vertex. Using this algorithm and Monte Carlo simulations, one obtains an efficient method to modify a unit ball graph by moving a single vertex to a location which maximizes the reliability. Another important consideration in many swarm missions is area coverage, yet highly reliable ball graphs often contain clusters of vertices. Here, we generalize our previous algorithm to improve area coverage as well as reliability. Our algorithm determines a location to add or move a vertex within a unit ball graph which maximizes the reliability, under the constraint that no other vertices of the graph be within some fixed distance. We compare this method of obtaining graphs with high reliability and evenly distributed area coverage to another method which uses a modified Fruchterman-Reingold algorithm for ball graphs.
\end{abstract}


\section{Introduction}
\IEEEPARstart{R}{ecent} decades have seen a rise in autonomous swarms of agents of various types. For instance, autonomous swarms of satellites are likely to replace monolithic NASA missions in the near future, due increased flexibility and reduced cost~\cite{morgan2012swarm, stoica, taxonomy, fishman}.
Swarms of autonomous underwater vehicles for oceanography and coastal management~\cite{9153840}, unmanned aerial vehicles for agricultural and environmental purposes~\cite{budiharto2019review, saffre2023wild}, and swarm robotics for search and rescue missions~\cite{hu2020voronoi, dorigo2021swarm} (among other purposes) are also the subject of significant current research.
In order that the agents in a swarm work together to accomplish a task, such as Earth observation by satellites~\cite{farrag2021satellite}, each is equipped with a communication device.
Without centralized control, the network induced by the agents and their communication links must remain connected; that is, any agent should be able to pass a message to any other, possibly via some intermediary agents.
Otherwise, were the network to become disconnected, the swarm would not be able to function as a cohesive unit.
As connectedness of the network is crucial, and as various uncertainties may render agents or the communication links between them inoperable, it is of great importance to design formations which induce communication networks whose connectedness is robust to potential failures.
In the context of satellite swarms, for instance, such uncertainties may include hardware issues, radiation, space debris, or, in a lower earth orbit, clouds and other atmospheric conditions~\cite{chen2024reliability, ESA, Howell, SoA}.

Assuming that every agent in a swarm is equipped with the same omnidirectional communication device, the swarm's communication network can be modeled by a geometric graph known as a {\em unit ball graph}.
Vertices of the ball graph are points in space (representing agents) and edges connect vertices within distance $1$ (communication range after scaling).
Unit ball graphs are also commonly used to model a wide variety of other communication networks, especially those using radio broadcast or optimal communication.
The two-dimensional counterpart to a unit ball graph is known as a {\em unit disk graph}.
These are often easier to work with and, in a number of applications such as radio broadcast networks~\cite{aboelfotoh_computing_1989}, reasonable simplifications.
In an application to satellite imaging of a region on the surface of an object, it is also reasonable to assume that the satellites in a formation lie approximately in a plane above the region.

In order to measure the robustness of the connectedness of a network, a theory known as network reliability has sprung forth.
Commonly, one assigns a probability of operation to each edge and vertex of a graph, and asks for the probability of connectedness of the subgraph obtained by including each vertex independently with its given probability of operation, and each edge between operational vertices independently with its given probability of operation. 
The case in which all vertices, but not all edges, operate with probability $1$ is known as {\em all-terminal reliability}~\cite{colbourn_combinatorics_1987}.
We refer to the more general case simply as {\em reliability}, though in the literature it has been referred to as {\em residual connectedness}~\cite{shpungin_combinatorial_2006}.
Unfortunately, the all-terminal reliability, even of a unit disk graph, is $\sharp${\sf P}-complete to compute~\cite{aboelfotoh_computing_1989}.
A large amount of work has thus gone into efficiently estimating~\cite{dagum2000optimal} and providing bounds~\cite{ball_chapter_1995} on the reliability of a graph.

Now, in the context of a hypothetical mission for an autonomous swarm, our goal is to ensure that the communication network has high reliability.
Of course, a formation in which all agents are pairwise within communication range is optimally reliable, but mission constraints may render such a formation inefficient.
For instance, suppose that a satellite swarm is assigned to image a region on the surface of an object.
It is natural to assign some satellites to image the outer boundary of the region.
The problem is then to assign the remaining satellites to cover the interior of the region so that the entire region is evenly covered and so that the communication network is highly reliable.

One method to keep agents spread out is to ensure some minimum distance between any given pair.
The main contribution of this paper is an efficient algorithm to enumerate all possible neighborhoods that a vertex $v$ could have when added to a unit disk graph, under the constraint that no other vertices are within some fixed Euclidean distance $b$ of $v$ (Algorithm~\ref{alg:enumerate_buffer_regions} of Section~\ref{sec:buffer}).
We use this algorithm to move vertices one at a time, under the same constraint, to locations which maximize the reliability of the resulting graph.
An alternative method is to use a spring layout- (or Fruchterman-Reingold)-type algorithm, treating vertices like similarly charged particles which repel and edges like springs which attract or repel based on the resting length; we provide a modified spring layout algorithm which avoids destroying edges in a unit ball graph in Section~\ref{sec:spring}.

The remainder of this paper is organized as follows.
First, in Section~\ref{sec:radial_and_ball}, we review an algorithm from our preliminary conference paper~\cite{buchanan2023node} to enumerate all possible neighborhoods for a vertex added to a unit ball graph.
This can be used to move a vertex in a unit ball graph to a location which maximizes the network reliability of the resulting graph.
In Section~\ref{sec:area_coverage}, we consider both reliability and area coverage.
In Section~\ref{sec:buffer}, we describe our main algorithm in terms of unit disk graphs, noting that this algorithm can be generalized to unit ball graphs as well.
In Section~\ref{sec:spring}, we provide a variation of a classical algorithm of Fruchterman and Reingold~\cite{fruchterman1991graph} to space out vertices in a ball graph without destroying edges.
In Section~\ref{sec:formation_planning}, we compare the effectiveness of the algorithms in Sections~\ref{sec:buffer} and~\ref{sec:spring} in producing unit disk graphs with high reliability and evenly spread vertices over a region.
The comparison shows the former algorithm to be particularly effective.
We conclude with some open questions and future directions in Section~\ref{sec:future_directions}.

\subsection{Definitions and notations}

Let $V$ be a set of points in space, and let $G$ be the graph with vertex set $V$ whose edges connect pairs of distinct vertices $u,v$ with $\|u - v\|_2 < 1$.
If $V \subset \mathbb{R}^2$, we call $G$ a {\em unit disk graph}, and if $V \subset \mathbb{R}^3$, we call $G$ a {\em unit ball graph}.
The {\em neighborhood} of a vertex $v$ in $V$ is the set of vertices adjacent to $v$.

Let $v$ be a point in $\mathbb{R}^2$. For positive real numbers $b$ and $c$ with $b < c$, we define the {\em $(b,c)$-annulus} centered at $v$ to be the set of points $w$ in $\mathbb{R}^2$ such that $b < \|v - w\|_2 < c$; its boundary is the set of points $w$ such that $\|v - w\|_2 \in \{b,c\}$.

Let $G$ be a graph, and for each edge $e$ in $G$, assign a probability of operation $p_e \in [0,1]$.
The {\em all-terminal reliability} of $G$, denoted $\Rel_A(G)$, is the probability that the graph $G'$ obtained from $G$ by deleting each edge $e$ independently with probability $1-p_e$ is connected.
In other words, $G'$ is the graph on the vertex set of $G$ obtained by sampling each edge $e$ in $G$ independently with probability $p_e$.
If each vertex $v$, along with its incident edges in $G'$, is subsequently deleted independently with probability $1 - p_v$, the probability that the resulting graph $G''$ is connected is called the {\em residual connectedness} $\Rel(G)$.
Equivalently, $G''$ is obtained by first sampling each vertex $v$ in $G$ independently with probability $p_v$, then sampling the edges $e$ between operational vertices independently with probability $p_e$.

\subsection{Problem statement}
Let $R$ be a connected region in $\mathbb{R}^2$, $V$ a finite set of points in $R$, and $G$ the unit disk graph with vertex set $V$.
Our goal is to modify $G$ by moving or adding vertices within $R$ in order to increase $\Rel(G)$ or $\Rel_A(G)$ while decreasing the size of the largest empty circle (containing no points in $V$) whose center is contained in $R$.
We do so in order to determine highly reliable unit disk graphs in $R$ which do not leave large uncovered areas.
We note that deleting a vertex from $G$ can increase the reliability without increasing the size of the largest empty circle.
While our results are easily adapted to accommodate this case, our application to formation planning of autonomous swarms, in which agents are not typically dispensable, makes this an undesirable option.

\section{Radial sweeps and ball graphs}\label{sec:radial_and_ball}

We begin with preliminaries on radial sweep algorithms as they apply to disk and ball graphs.
We outline an algorithm from our preliminary paper~\cite{buchanan2023node}, using radial sweeps to list all possible neighborhoods that a new vertex could have when added to a disk or ball graph.
This will provide a basis for our main algorithm, described in Section~\ref{sec:area_coverage}.

\subsection{Radial Sweeps}\label{sec:radial}

Radial sweeps are commonplace tools in computational geometry used to determine, for example, the maximum number of points from a finite set $V \subset \mathbb{R}^2$ which are contained in a circle of radius $1$~\cite{jindal2023}.
The idea is quite simple: we imagine a unit circle with center $v + \left[ \begin{matrix} 1 & 0 \end{matrix} \right]^\T$ for a given point $v$ in $V$.
We rotate, or {\em sweep}, the circle $360$ degrees counterclockwise, keeping $v$ fixed on its boundary.
The angles of rotation at which each point $w$ within distance $2$ of $v$ enters and exits the circle as it sweeps can be calculated using basic trigonometry.
In particular, if $v = \left[ \begin{matrix} x_v & y_v \end{matrix} \right]^\T$ and $w = \left[ \begin{matrix} x_w & y_w \end{matrix} \right]^\T$, then letting $x = x_v - x_w$, $y = y_v - y_w$, $A = \arctan{(y / |x|)}$, and $B = \arccos{(\sqrt{x^2 + y^2} / 2)}$, the vertex $w$ enters the sweep around $v$ after a rotation of $\pi - (A + B)$ degrees and exits at $\pi - (A - B)$ degrees.
See Figure~\ref{fig:angularsweep}, adapted from~\cite{buchanan2023node}.

\begin{figure}
    \centering
    \begin{tikzpicture}
    \draw[->] (0,0) -- (4,0);
    
    \draw (90/3 - 10 : 1.5) circle (1.5);
    \draw[dashed, ->] (0,0) -- (90/3 - 10 : 4);

    \draw[black!70] (90 - 10 : 1.5) circle (1.5);
    \draw[black!70, dashed, ->] (0,0) -- (80:4);

    \node (theta0) at (8 : 4.2) {$\tmin$};

    \node (theta1) at (90/1.75 - 5 : 3.8) {\textcolor{black!60}{$\tmax$}};

    \draw[fill=black!100, draw=black!100] (0,0) circle (.05);
    \node (p) at (90/3 : -.25) {$v$};

    \draw[fill=black!100, draw=black!100] (90/3 - 10 : 1.5) + (80:1.5) circle (.05);
    \node (q) at (1.5, 1.75) {$w$};

    \draw[black!70, ->] (0 : 3.5) arc (0 : 80 : 3.5);

    \draw[black!100, ->] (0 : 3.75) arc (0 : 20 : 3.75);
    
    \end{tikzpicture}
    \caption{A circle of radius $1$ revolves in a counterclockwise motion around a fixed point $v$ on its boundary in a radial sweep.
    A point $w$ enters and exits the circle at the angles $\tmin$ and $\tmax$, respectively. (Adapted from~\cite{buchanan2023node}.)}
    \label{fig:angularsweep}
\end{figure}
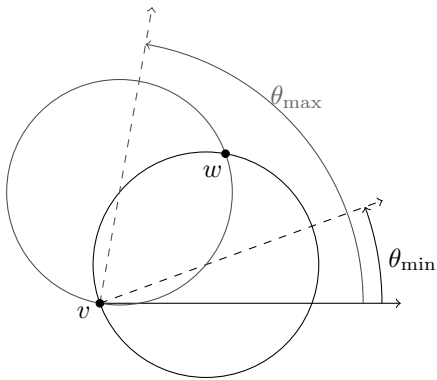

Performing these trigonometric calculations for each vertex $v$ in $V$ and each vertex $w$ in $V - v$ with $\|v - w\|_2 < 2$, one can determine all subsets $N$ of $V$ for which some radius-$1$ circle contains all of the points in $N$ and none in $V - N$ as follows.
Each time a point $w \in V$ enters or exits the circle sweeping around a point $v \in V$, we record the subset $S$ of $V$ contained in the interior this circle with $v$ and $w$ on its boundary.
By shifting and/or rotating the circle slightly, one can obtain a unit circle with any of $S$, $S \cup v$, $S \cup w$, or $S \cup \{v,w\}$ in its interior.
All of these are subsets $N$ of $V$ as described.
On the other hand, if $N$ is the subset of $V$ contained in a given radius-$1$ circle, then the circle can be shifted until two points in $V$ lie on its boundary.
Thus, after performing a radial sweep around each vertex in $V$, we obtain in this manner all such subsets $N$ of $V$.

Note that a subset $N$ of $V$ contained in the interior of some radius-$1$ circle which contains no other points in $V$ would be the neighborhood of a vertex placed at the center of said circle in the corresponding unit disk graph. Conversely, any subset of $V$ which could be the neighborhood of an appropriately placed vertex is contained in a radius-$1$ circle which contains no other points from $V$.
Thus, the previously summarized algorithm finds all possible neighborhoods for a new vertex added to a unit disk graph on $V$.
Although the number of possible neighborhoods for a new vertex added to an abstract graph is exponential, the number of such subsets $N$ of $V$ is on the order of $n^2$, where $n = |V|$~\cite{yaglom1964challenging}.
The algorithm described above is nearly best possible, running in $\mathcal{O}(n^2)$ time.
Further, our algorithm generalizes to finite subsets $V$ of $\mathbb{R}^3$, running in $\mathcal{O}(n^3)$ time, which is also the order of the maximum number of subsets $N$ of $V$ contained in a unit sphere with no other points from $V$ in its interior~\cite{yaglom1964challenging}.
These upper bounds can be improved using parameterized complexity. 
Letting $\Delta$ denote the maximum number of points in $V$ contained in a ball of radius $2$, our 2-D algorithm runs in time $\mathcal{O}(n\Delta)$ and the 3-D variation in time $\mathcal{O}(n^2\Delta)$~\cite{buchanan2023node}.

\subsection{Unit disk and unit ball graphs}

Disk and ball graphs are used to model a wide array of spatial relationships between objects.
Supposing that each agent in an autonomous swarm is equipped with the same omnidirectional communication device, we can model the swarm formation by a unit ball graph\footnote{We scale the space by a factor of $1/r$, where $r$ is the range of the communication device.} whose vertices represent agents in the formation and whose edges connect pairs of agents within communication distance.

Given the dynamic nature of swarm formations, it is natural to ask how the communication network changes with the formation.
Generally, one might like to know which changes to the graph structure of a unit ball graph $G$ are obtainable by repositioning vertices.
This problem is difficult in general, as the recognition problem for unit ball graphs is {\sf NP}-hard~\cite{breu1998unit}.
However, using the algorithm described in Section~\ref{sec:radial}, we are able to enumerate in polynomial time all of the unit ball graphs obtainable by moving a single vertex in $G$.
This can be modeled by deleting a vertex $v$ from $G$ and adding it back in a new location.
The algorithm summarized in Section~\ref{sec:radial} can be used to list all of the possible neighborhoods that $v$ could have.

\section{Reliability and area coverage}\label{sec:area_coverage}

Equipped with an algorithm to efficiently determine all possible neighborhoods for a new vertex added to a unit disk or unit ball graph, in~\cite{buchanan2023node}, we used Monte Carlo simulations to determine the location to add or move a vertex which would maximize either $\Rel$ or $\Rel_A$.
Indeed, for any efficiently calculable or estimable graph parameter $P$, one can use our algorithm to optimize $P$ by moving or adding a single vertex.
We note that, when using our algorithm to maximize reliability, only those neighborhoods which are maximal with respect to set inclusion need be considered, for $\Rel$ and $\Rel_A$ are both monotonic with respect to deleting edges.
We compared the reliabilities, over all maximal neighborhoods for a new vertex, of the graphs obtained by adding a vertex with said neighborhood.
For the neighborhood $N$ providing the most reliable estimate, we added a new vertex at the center of the smallest enclosing circle for $N$, which is the point which minimizes the maximum distance to a point in $N$.
Finding this center can be done in linear time~\cite{megiddo1983linear, welzl2005smallest}.


We tested this method on random geometric graphs in~\cite{buchanan2023node}, finding that repositioning a single vertex using our algorithm provides higher all-terminal reliability on average than adding ten additional vertices uniformly at random to the visible area ({\em i.e.}, adding vertices randomly conditioned on connectedness) with edge-operation probabilities of $90\%$.
Applying our algorithm to add new vertices one at a time to maximize reliability, we notice that, after a number of iterations, the vertices are clustered together.
This aligns with the results of~\cite{krupke_distributed_2015-1} in which highly reliable graphs were constructed to connect a fixed set of points, and these resembled Steiner trees with thick bands of vertices.

In the context of many a swarm mission, like satellite imaging of a region, such a formation may not be desirable.
A natural question then arises: {\em which unit disk graphs have evenly distributed vertices over a given area and have high reliability?}
In the spirit of our first algorithm, we propose to add or move vertices one at a time, maximizing reliability at each iteration, but imposing the constraint that no vertex be within some fixed distance $b$ of the vertex being added or moved.
We present such an algorithm in Section~\ref{sec:buffer}, and we propose an alternative solution to the area coverage problem using a modified spring layout algorithm in Section~\ref{sec:spring}.

\subsection{Finding neighborhoods for a vertex with a buffer}\label{sec:buffer}

The methods we provided in~\cite{buchanan2023node} were not geared toward a specific autonomous swarm mission, and we acknowledge that high reliability, while important, is not the only quality desired of a swarm formation in performing a mission.
A tight cluster of satellites, for instance, has a highly reliable communication network, but would not perform a task such as imaging a large area very well, where maximizing the spread of the satellites would maximize the area that can be imaged.
One can imagine similar scenarios in other spatial networks, such as swarm robotic search and rescue missions.
It thus seems natural to avoid tight clustering of vertices while keeping reliability high.

There are a number of ways that one might handle this problem.
We propose a method whereby a buffer distance is imposed around each vertex as it is added to the unit disk graph $G$.
For a value $b$ with $0 < b < 1$, we note that a vertex which has no neighbors within distance $b$ must have all of its neighbors contained in the $(b,1)$-annulus centered at that vertex (see Figure~\ref{fig:vx_w/buffer}).
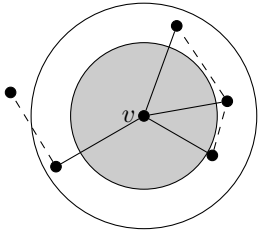
\begin{figure}
    \centering
    \begin{tikzpicture}
    [every node/.style={circle, draw=black!100, fill=black!100, inner sep=0pt, minimum size=4pt}, scale=1.5]

    \draw[fill=none] (0,0) circle (1);
    \draw[fill=black!20] (0,0) circle (.65);
    \node (v) [label={left:$v$}] at (0,0) {};

    \node (a) at (10:.75) {};
    \node (b) at (70:.85) {};
    \node (c) at (-30:.7) {};
    \node (d) at (-150:.9) {};
    \node (x) at (170:1.2) {};

    \draw[dashed] (c)--(a)--(b);
    \draw[dashed] (d)--(x);

    \foreach \i in {a,b,c,d}
    {
    \draw (v)--(\i);
    }
    
    \end{tikzpicture}
    \caption{The neighbors of a vertex $v$ added with a buffer to a disk graph are contained within an annulus with no vertices in its inner circle. The buffer is depicted by a shaded region, and existing edges in the graph before the addition of $v$ are depicted by dashed edges.}
    \label{fig:vx_w/buffer}
\end{figure}
In other words, the possible neighborhoods for a new vertex added with a buffer are precisely the subsets $S$ of the vertex set $V$ of $G$ such that $S$ is contained in a $(b,1)$-annulus and all points in $V - S$ lie outside of the radius-$1$ circle.
We use a modified radial sweep algorithm to determine all such subsets $S$ of $V$.
This subsection is devoted to proving the following.

\begin{theorem}\label{thm:buffer}
    Let $G$ be a unit disk graph with vertex set $V$, and let $0 < b < 1$.
    The list of possible neighborhoods for a vertex $v$ added to $G$ such that $\|v - w\|_2 > b$ for all $w \in V$ can be obtained in $\mathcal{O}(n \Delta)$ time, where $\Delta$ is the maximum cardinality of a subset of $V$ contained in a radius-$2$ circle.
\end{theorem}

Theorem~\ref{thm:buffer} follows from an analysis of Algorithm~\ref{alg:enumerate_buffer_regions}, 
which we describe in parts throughout this section.
We first find all subsets contained in a $(b,1)$-annulus with two vertices on its boundary and no vertices contained within its inner circle.
To do so, we conduct radial sweeps of annuli around a fixed points on their outer and inner boundaries.

\begin{algorithm}
\caption{Enumerate neighborhoods for a new vertex added to a unit disk graph with no other vertices within distance $b$ ($0 < b < 1$).}
\label{alg:enumerate_buffer_regions}
\begin{algorithmic}[1]
    \Require {Nonempty finite set of points $V \subset \mathbb{R}^2$ and $b \in (0,1)$}
    \Ensure {Subsets of $V$ contained in a $(b,1)$-annulus with no points in inner circle}
    \Statex
    
    \Function{BufferNeighborhoods}{$V,b$}
    \Let{$\mathrm{Nbrhds}$}{$\{\}$}
    \For{$v \in V$}
        \Let{$L_1$}{$\textsc{OuterSweep}(V, v, b)$}
        \For{$S \in \textsc{SetsFromIntervals}(L_1)$}
            \State{add $S$ and $S \cup v$ to $\mathrm{Nbrhds}$}
        \EndFor
        \Let{$L_2$}{$\textsc{InnerSweep}(V,v,b)$}
        \For{$S \in \textsc{SetsFromIntervals}(L_2)$}
            \State{add $S \cup v$ to $\mathrm{Nbrhds}$}
        \EndFor
    \EndFor
    \State \Return{$\mathrm{Nbrhds}$}
    \EndFunction
\end{algorithmic}
\end{algorithm}


\begin{algorithm}
\caption{List of points and angles at which points enter and exit a $(b,1)$-annulus swept around a fixed point $v$ on its outer boundary}
\label{alg:outer_sweep}
\begin{algorithmic}[1]
    \Require{Nonempty finite set of points $V \subset \mathbb{R}^2$, $v \in V$, and $b \in (0,1)$}
    \Ensure{Points $w \in V$ and angles $\theta \in (-\pi/2, \pi/2)$ of rotation at which $w$ enters and exits a $(b,1)$-annulus swept counterclockwise around fixed point $v$ on outer boundary
    }
    \Statex

    \Function{OuterSweep}{$V,v,b$}
    \Let{$L$}{$\{ \}$}
    \For{$w \in V - v$ with $\|v - w\|_2 < 2$}
    \State
    \Comment{$w$ enters \& exits the sweep}
        \If{$1 - b < \|v - w\|_2 < 1 + b$}
        \State 
        \Comment{$w$ enters \& exits both circles}
            \Let{$\theta_{\mathrm{min}}$}{angle $w$ enters outer circle}
            \Let{$\tdis$}{angle $w$ enters inner circle}
            \Let{$\tre$}{angle $w$ exits inner circle}
            \Let{$\tmax$}{angle $w$ exits outer circle}
        \Else
        \State
        \Comment{$w$ only enters \& exits the outer circle}
            \Let{$\tmin$}{angle $w$ enters outer circle}
            \Let{$\tdis$, $\tre$}{None}
            \Let{$\tmax$}{angle $w$ exits outer circle}
        \EndIf
      \State{add $(w, \tmin, \tdis, \tre, \tmax)$ to $L$}
    \EndFor
        
    \State \Return{ $L$ }
    \EndFunction
\end{algorithmic}
\end{algorithm}

\begin{algorithm}
\caption{List of points and angles at which they enter and exit a $(b,1)$-annulus swept around a fixed point on its inner boundary}
\label{alg:inner_sweep}
\begin{algorithmic}[1]
    \Require{Nonempty finite set of points $V \subset \mathbb{R}^2$, $v \in V$, and $b \in (0,1)$}
    \Ensure{Points $w \in V$ and angles $\theta \in (-\pi/2, \pi/2)$ of rotation at which $w$ enters and exits a $(b,1)$-annulus swept counterclockwise around fixed point $v$ on inner boundary}
    \Statex

    \Function{InnerSweep}{$V,v,b$}
    \Let{$L$}{$\{ \}$}
    \For{$w \in V - p$ with $\|v - w\|_2 < 1+b$}
    \State
    \Comment{$w$ enters \& exits the sweep}
        \If{$1-b < \|v - w\|_2 < 2b$}
        \State 
        \Comment{$w$ enters \& exits both circles}
            \Let{$\tmin$}{angle $w$ enters outer circle}
            \Let{$\tdis$}{angle $w$ enters inner circle}
            \Let{$\tre$}{angle $w$ exits inner circle}
            \Let{$\tmax$}{angle $w$ exits outer circle}
        \ElsIf{$\|v - w\|_2 > 2b$}
        \State 
        \Comment{$w$ only enters \& exits outer circle}
            \Let{$\tmin$}{angle $w$ enters outer circle}
            \Let{$\tdis$, $\tre$}{None}
            \Let{$\tmax$}{angle $w$ exits outer circle}
        \ElsIf{$\|v - w\|_2 < 1-b$}
        \State 
        \Comment{$w$ only enters \& exits inner circle}
            \Let{$\tmin$, $\tmax$}{None}
            \Let{$\tdis$}{angle $w$ enters inner circle}
            \Let{$\tre$}{angle $w$ exits inner circle}
        \EndIf
    \State{add $(w, \tmin, \tdis, \tre, \tmax)$ to $L$}
    \EndFor
        
    \State \Return{ $L$ }
    \EndFunction
\end{algorithmic}
\end{algorithm}

First, we sweep a $(b,1)$-annulus around a fixed point $v \in V$ on its outer boundary (see Algorithm~\ref{alg:outer_sweep}).
Consider the $(b,1)$-annulus whose center is at $v + \left[ \begin{matrix} 1 & 0 \end{matrix} \right]^\T$.
We imagine a full counterclockwise rotation of the annulus around the fixed point $v$.
Points $w \in V$ with $1 - b < \|v - w\|_2 < 1 + b$ will enter and exit both the inner and outer circles of the annulus as it rotates.
Using trigonometric equations similar to those described in Section~\ref{sec:radial}, we record the angles of rotation for each of these four moments: let $\tmin$ denote the angle at which $w$ enters the outer circle, $\tdis$ the angle at which $w$ enters the inner circle (and disappears from the annulus), $\tre$ the angle at which $w$ exits the inner circle (and reappears in the annulus), and $\tmax$ the angle at which $w$ exits the outer circle.
We add the $5$-tuple $(w, \tmin, \tdis, \tre, \tmax)$ to a list $L$.
Points $w \in V$ with $\|v - w\|_2 < 1-b$ or $\|v - w\|_2 > 1 + b$ will enter and exit only the outer circle of the annulus as it rotates about $v$.
In this case, we set $\tdis = \mathrm{None}$ and $\tre = \mathrm{None}$.

Second, we sweep a $(b,1)$-annulus around a fixed point $v \in V$ on its inner boundary (see Algorithm~\ref{alg:inner_sweep}).
Consider the $(b,1)$-annulus centered at $v + \left[ \begin{matrix} b & 0 \end{matrix} \right]^\T$.
As it rotates counterclockwise about $v$, points $w \in V$ within distance $1 + b$ of $v$ will enter and exit the annulus.
If $1 - b < \|v - w\|_2 < 2b$, then $w$ will enter and exit the annulus twice in one full rotation, passing through both the inner and outer boundaries.
We again record the angles $\tmin$, $\tdis$, $\tre$, and $\tmax$ defined above, using basic trigonometry, and add the tuple $(w, \tmin, \tdis, \tre, \tmax)$ to a list $L$.
If $2b < \|v - w\|_2 < 1+b$, then $w$ only enters and exits the annulus once during the sweep, via the outer boundary.
We record the angles $\tmin$ and $\tmax$ and set $\tdis = \mathrm{None}$ and $\tre = \mathrm{None}$.
Finally, if $\|v - w\|_2 < 1-b$, then $w$ only enters and exits the annulus via the inner boundary, and we set $\tmin = \mathrm{None}$ and $\tre = \mathrm{None}$.

\begin{algorithm}
    \caption{Obtain sets of points contained in a $(b,1)$-annulus from a sweep around a given point $v$}
    \label{alg:sets_from_intervals}
    \begin{algorithmic}[1]
        \Require{Nonempty finite subset $V$ of $\mathbb{R}^2$; $L$, the output of Algorithm~\ref{alg:outer_sweep} or Algorithm~\ref{alg:inner_sweep}}
        \Ensure{Collection of subsets of $V$ contained in a $(b,1)$-annulus with no points in its inner circle}
        \Statex

        \Function{SetsFromIntervals}{$L$}
        \Let{$M$}{$[\ ]$}
        \Let{$R$}{$\textsc{InitialSet}(L)$} 
        \State
        \Comment{See Algorithm~\ref{alg:initial_set} in Appendix~\ref{appendix:InitialSet}}
        \Let{$B$}{$\{\}$} 
        \State
        \Comment{Bad vertices: contained in the inner circle}
        \For{$(w, \tmin, \tdis, \tre, \tmax) \in L$ and }
            \For{$\theta \in \{\tmin, \tdis, \tre, \tmax\}$}
                \If{$\theta$ is not None}
                    \State{add $(w, \theta, \mathrm{type})$ to $M$, where $\mathrm{type} \in \{\min, \mathrm{dis}, \mathrm{re}, \max\}$}
                \EndIf
            \EndFor
            \If{$\tre < \tdis$} 
            \State
            \Comment{$w$ is bad when sweep begins}
                \State{add $(w, \tdis, \mathrm{type}=\mathrm{dis})$ to $B$}
            \EndIf
        \EndFor
        \State{Sort $M$ by angles $\theta$ in increasing order}
        \Let{$\mathrm{Sets}$}{$\{\}$}
        \For{$(w, \theta, \mathrm{type}) \in M$}
            \If{$\mathrm{type}=\max$}
                remove $w$ from $R$
                \If{$B = \emptyset$}
                    add $R$ and $R \cup w$ to $\mathrm{Sets}$
                \EndIf
            \ElsIf{$\mathrm{type}=\min$}
                add $w$ to $R$
                \If{$B = \emptyset$}
                    add $R$ and $R - w$ to $\mathrm{Sets}$
                \EndIf
            \ElsIf{$\mathrm{type}=\mathrm{dis}$}
                remove $w$ from $R$
                \If{$B = \emptyset$}
                    add $R \cup w$ to $\mathrm{Sets}$
                \EndIf
                \State{add $w$ to $B$}
            \ElsIf{$\mathrm{type}=\mathrm{re}$}
                add $w$ to $R$
                \State{remove $w$ from $B$}
                \If{$B = \emptyset$}
                    add $R$ to $\mathrm{Sets}$
                \EndIf
            \EndIf
        \EndFor
        \State \Return{$\mathrm{Sets}$}
        \EndFunction
    \end{algorithmic}
\end{algorithm}

We now describe how to obtain the sets $S$ of vertices contained in a $(b,1)$-annulus with $V - S$ outside of its outer boundary (see Algorithm~\ref{alg:sets_from_intervals}).
By Proposition~\ref{prop:vxs_on_boundary}, it suffices to find the sets $S$ contained in a $(b,1)$-annulus with two points on its boundary (keeping track of whether they are on the outer or inner boundary) and all other points in $V - S$ outside of the outer boundary.
One of the points on the boundary will be the point $v$ we are sweeping around, and the other will be a point that is entering or exiting the sweep.

Let $L_v$ be a list of tuples $(w, \tmin, \tdis, \tre, \tmax)$ obtained from a sweep around $v \in V$ using either Algorithm~\ref{alg:outer_sweep} or~\ref{alg:inner_sweep}.
We first determine the points contained in the annulus when the sweep begins (Algorithm~\ref{alg:initial_set}).
These are those points $w$ which are close enough to $v$ to enter and exit the sweep and are such that 
\begin{itemize}
    \item $\tmax < \tmin$ and $\tdis < \tre$, or 
    \item $\tmax < \tmin$ and $\tdis, \tre = \mathrm{None}$, or 
    \item $\tmin, \tmax = \mathrm{None}$ and $\tdis < \tre$.
\end{itemize}
We add these points to a ``running set" $R$.
Any points $w$ which are not contained in the annulus when the sweep begins but which are contained in its inner circle (those with $\tre < \tdis$) we add to a ``bad set" $B$.

From $L_v$, we make a new list $M_v$ of tuples $(w, \theta, \mathrm{type})$ to record the angles $\theta$, in increasing order, at which points $w$ enter and exit the annulus during the sweep, along with whether it is entering or exiting and which boundary it is passing through.
Specifically, $\mathrm{type}$ is one of `$\min$', `$\mathrm{dis}$', `$\mathrm{re}$', or `$\max$' depending on whether $\theta$ is the angle at which $w$ enters the outer circle, enters the inner circle, exits the inner circle, or exits the outer circle of the annulus, respectively.

Now, we go through the tuples in $M_v$ in increasing order of $\theta$.
If a vertex $w$ exits the outer circle (that is, if the tuple $(w, \theta, \mathrm{type})$ has type `$\max$'), we remove $w$ from $R$. We then check if the set of bad vertices $B$ is empty and, if it is, add both $R$ and $R \cup w$ to a list $\mathrm{Sets}$ which will be the output of Algorithm~\ref{alg:sets_from_intervals}.
If a vertex $w$ enters the outer circle (if the tuple has type `$\min$'), we add it to $R$, and if $B = \emptyset$, we add $R$ and $R - w$ to $\mathrm{Sets}$.
If a vertex $w$ enters the inner circle (if the tuple has type `$\mathrm{dis}$'), we remove $w$ from $R$, check if $B = \emptyset$, and, if it is, add $R \cup w$ to $\mathrm{Sets}$.
Then, we add $w$ to $B$.
Finally, if a vertex $w$ exits the inner circle (if the tuple has type `$\mathrm{re}$'), we add it to $R$ and remove it from $B$.
If now $B = \emptyset$, we add $R$ to $\mathrm{Sets}$.

Once the list $M_v$ is exhausted, $\mathrm{Sets}$ contains all sets of vertices which are contained in a $(b,1)$-annulus with $v$ fixed on a boundary.
If $v$ is fixed on the outer boundary, then for each set $S$ in $\mathrm{Sets}$, both $S$ and $S \cup v$ are possible subsets of $V$ contained in a $(b,1)$-annulus with no points in its inner circle.
On the other hand, if $v$ was fixed on the inner boundary, then only $S \cup v$ is a possible subset of $V$ contained in a $(b,1)$-annulus with no points in its inner circle (see Proposition~\ref{prop:vxs_on_boundary}).

Performing Algorithms~\ref{alg:outer_sweep} and~\ref{alg:inner_sweep} each $n = |V|$ times, by Proposition~\ref{prop:vxs_on_boundary} we obtain all subsets $S$ of $V$ for which there exists a $(b,1)$-annulus containing $S$ and with every point in $V - S$ outside of the radius-$1$ circle in $\mathcal{O}(n^2)$ time.
This is summarized in Algorithm~\ref{alg:enumerate_buffer_regions}.

As with our previous algorithm in~\cite{buchanan2023node}, the run time of Algorithm~\ref{alg:enumerate_buffer_regions} can be improved using parameterized complexity.
Letting $\Delta$ denote the maximum number of points contained in a radius-$2$ circle (the maximum number of neighbors that a vertex added to $G$ can have), our algorithm runs in time $\mathcal{O}(n \Delta)$.

We also note that our algorithms can be generalized to unit ball graphs in $\mathbb{R}^3$ by sweeping two concentric spheres, one of radius $b$ and one of radius $1$, around the axis formed by the line segment between two vertices $u, v \in V$ with $\|u - v\|_2 < 2$.
We do not describe such an algorithm here, as it would be a lengthy description, and as unit disk graphs are reasonable models for many of the relevant applications mentioned here (including our running example of an autonomous satellite swarm imaging a region on the surface of an object).

By adjusting the value of $b$, we can use this algorithm to fill in a region with a reliable graph whose vertices are evenly spaced apart.
In Section~\ref{sec:formation_planning}, we will compare this algorithm to another method for spreading out vertices without decreasing reliability which we presently describe.

\subsection{Fruchterman-Reingold for ball graphs}\label{sec:spring}

We now describe an alternative approach for modifying a unit ball graph for reliability and area coverage purposes.
In particular, we investigate a method for providing more even area coverage without decreasing the reliability of the graph obtained using Algorithm~\ref{alg:enumerate_buffer_regions} or the simpler version from~\cite{buchanan2023node} described in Section~\ref{sec:radial}.

The Fruchterman-Reingold, or spring layout, algorithm introduced in~\cite{fruchterman1991graph} is a standard graph drawing method.
It is based on a simple idea: vertices should be reasonably spaced apart, and adjacent vertices should be closer together than pairs of nonadjacent vertices.
To achieve this, vertices are treated like similarly charged particles, repelling one another, while edges are treated like springs, resisting this force.
The algorithm iterates some number of times, at each iteration linearly scaling down this combination of attractive and repulsive forces (reducing the ``temperature") until vertices settle into place.

Spring layout algorithms have also been used in a variety of applications.
Variations of the Fruchterman-Reingold algorithm have been proposed for the deployment of large-scale wireless sensor networks~\cite{deng2019optimized,li2022fruchterman}.
A related concept of virtual force algorithms has also been studied in the context of wireless sensor network coverage (see, for example,~\cite{howard2002mobile,qi2022wireless}).

While the spring layout algorithm can provide a more even spread of vertices over a region, it can also drastically change the underlying ball graph.
Of concern to us is the scenario in which edges are deleted (when vertices which were within communication distance are moved further apart), decreasing the reliability of the resulting ball graph.
Extra edges, on the other hand, only increase the reliability.
Luckily, this simple algorithm is not hard to modify so that it respects adjacencies in a unit disk or ball graph.
We chose to modify the existing Fruchterman-Reingold algorithm from the NetworkX Python library~\cite{SciPyProceedings_11}.\footnote{\url{https://networkx.org/documentation/stable/_modules/networkx/drawing/layout.html}}
There are a number of ways that one might stop edges in a unit ball graph from being broken by the Fruchterman-Reingold algorithm.
Here, we do so by leveraging the temperature: at each iteration, we check whether any pair vertices which were previously within distance $1$ are now of distance greater than $1$ apart.
If such a pair exists, we simply go back to the positions at the last iteration, reduce the temperature, and move on to the next iteration.
We now describe our modified Fruchterman-Reingold algorithm for ball graphs in more detail.

Let $G$ be a unit ball graph of order $n$ with vertex set $V$ and edge set $E$.
We consider $G$ to lie within a {\em frame}, a minimum rectangle or rectangular prism containing all points in $V$.
Let $k$ denote the optimal distance between an average pair of vertices at the end of the algorithm.
In two dimensions, it is common to set $k = \sqrt{A/n}$, where $A$ denotes the area of the frame.
Finally, let $t$ denote the initial {\em temperature}, which determines how far we allow vertices to move at the first iteration.
For the algorithm we chose to modify, $t$ is $1/10$ of the maximum length of a side of the frame.
For each pair of distinct vertices $u,v \in V$, we calculate the force $f_r(u,v) = k^2 / d_{u,v}^{2}$ of repulsion.
When $uv \in E$, we calculate the force $f_a(u,v) = d_{u,v} / k$ of attraction, and we let $f_a(u,v) = 0$ when $uv \notin E$.
Let $F$ denote the matrix of forces between vertices: $F_{u,v} = f_r(u,v) - f_a(u,v)$.
We use these forces to calculate the initial displacement vector $\Vec{\mathrm{disp}}_v$ for each vertex $v$.
In two dimensions, letting $v = \left[ \begin{matrix} x_v & y_v \end{matrix} \right]^\T$, the first coordinate of $\Vec{\mathrm{disp}}_v$ is $\sum_{w \neq v} (x_v - x_w) F_{v,w}$ and the second is $\sum_{w \neq v} (y_v - y_w) F_{v,w}$.
We scale each displacement vector according to the temperature, obtaining a vector $\delta_v = \Vec{\mathrm{disp}}_v \cdot t / l_v$, where $l_v$ is the length of $\Vec{\mathrm{disp}}_v$.

Classical Fruchterman-Reingold algorithms come in two types: at each iteration, either all vertices move at once according to their displacement vectors, or the vertices move one at a time, requiring a recalculation of forces at each iteration.
In the former case, the new position for each vertex $v$ is $\left[ \begin{matrix} x_v & y_v \end{matrix} \right]^\T + \delta_v$.
In the latter case, we cycle through the vertices of $G$, moving a single vertex $v$ according to $\delta_v$ at each iteration.
Our solution to the problem of breaking edges is essentially the same in either case: if any pair of vertices previously within distance $1$ would be at distance $> 1$ after any iteration, we avoid moving any vertices at all, lower the temperature by $t / (\mathrm{it} + 1)$, where `$\mathrm{it}$' denotes the total number of iterations, and move on to the next iteration (essentially weakening the forces $F_{u,v}$).
If all pairs of vertices previously within distance $1$ remain within distance $1$, we keep the new positions and again lower the temperature by $t / (\mathrm{it} + 1)$ for the next iteration.
As the temperature decreases to $t / (\mathrm{it} + 1)$ at the last iteration, the positions of the vertices stabilize.


\section{Formation planning for autonomous swarms}\label{sec:formation_planning}

In this section, we apply Algorithm~\ref{alg:enumerate_buffer_regions} and the modified Fruchterman-Reingold algorithm from Section~\ref{sec:spring} to obtain swarm formations with even area coverage and high reliability.
Returning to our running example, we consider an autonomous swarm of satellites assigned to image a region on the surface of an object.
We begin by assigning satellites to image the outer boundary of the region.
Our goal is to assign the remaining satellites to cover the interior of the region so that the entire region is evenly covered and so that the communication network is not only connected but highly reliable.
We note that search and rescue missions by robots, imaging of an ocean bed by autonomous underwater vehicles, or use of unmanned aerial vehicles for precision agriculture may be modeled similarly.
Our findings indicate that using Algorithm~\ref{alg:enumerate_buffer_regions} alone (as opposed to using it in conjunction with our modified Fruchterman-Reingold algorithm) is particularly effective in accomplishing this goal.

\begin{figure}
    \centering
    \includegraphics[scale=.45]{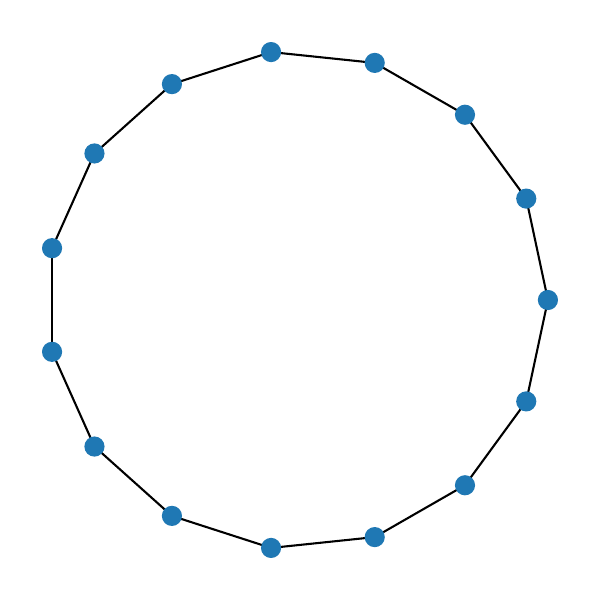}
    \caption{Fifteen agents cover the outer boundary of a circular region, modeled by a unit disk graph forming a regular $15$-gon whose edges are of length $0.9$}
    \label{fig:15-gon}
\end{figure}

As a basis for our simulations, we imagine a circular disk to be covered by a swarm of $30$ agents.\footnote{This is simply for illustration; regions of other shapes, and swarms of other sizes, may be considered as well.}
We begin with a unit disk graph on $15$ vertices forming a regular $15$-gon whose edges have length $0.9$ (see Figure~\ref{fig:15-gon}).
These vertices outline a circular region with area about $14.717$ (radius $\approx 2.16$).
We add $15$ additional vertices one at a time to this unit disk graph using four different methods:

\begin{enumerate}[(M1)]

    \item\label{item:rand}
    {\em Random:}
    Add vertices randomly to the interior of the region, conditioned on them being within unit distance of at least one existing vertex ({\em i.e.}, conditioned on connectedness).

    \item\label{item:rand_spring} 
    {\em Random \& Spring:}
    Add vertices as in item~(M\ref{item:rand}), 
    but apply the Fruchterman-Reingold algorithm from Section~\ref{sec:spring} each time a vertex is added, leaving the positions of the initial $15$ vertices fixed.
    
    \item\label{item:buffer} 
    {\em Buffer:}
    Use Algorithm~\ref{alg:enumerate_buffer_regions} and Monte Carlo simulations to add vertices which maximize the resulting reliability while being at distance at least $0.65$ from any other vertex.
    
    \item\label{item:buffer_spring} 
    {\em Buffer \& Spring:}
    Add vertices as in item~(M\ref{item:buffer}), but apply the Fruchterman-Reingold algorithm from Section~\ref{sec:spring} each time a vertex is added.
\end{enumerate}

The specifics of each of these methods are discussed below, and examples of the resulting graphs are depicted in Figure~\ref{fig:sim_pics}.
For a reliability measure, we use the all-terminal reliability with edge-operation probabilities all $90\%$.
These edge-operation probabilities may be adjusted, and unreliable vertices introduced, with additional information regarding the agents in the application and their communication devices.
Using the Zero-One Estimator theorem of Karp and Luby~\cite{karp_monte-carlo_1985}, or one of its improved forms (see, {\em e.g.},~\cite{dagum2000optimal}), one can easily compute upper bounds on the number of Monte Carlo simulations necessary to obtain, with probability $1 - \delta$, an estimate of $\Rel_A$ with relative error $\epsilon$ (known as an $(\epsilon, \delta)$-approximation).
For a graph $G$ with high reliability, however, these upper bounds can be quite large.
Such Monte Carlo techniques proved inefficient for the more reliable graphs produced by the methods in items~(M\ref{item:buffer}) and~(M\ref{item:buffer_spring}) above.
Instead, we use an open-source program called $\mathcal{K}$-RelNet, introduced in~\cite{paredes2019principled}, to obtain $(\epsilon, \delta)$-estimates of $\Rel_A(G)$.
Essentially, the program reduces the problem of finding all-terminal reliability (or the more general $K$-terminal reliability, see~\cite{colbourn_combinatorics_1987}) to counting the number of assignments satisfying a Boolean formula in conjunctive normal form.
It then obtains an $(\epsilon, \delta)$-approximation of this count using a state-of-the-art approximate model counter ApproxMC, introduced in~\cite{chakraborty2013scalable} and now in its sixth version, ApproxMC6~\cite{yang2023rounding}.
We set $\epsilon = 0.8$ and $\delta = 0.2$ for our relative estimates; these are the default values for the $\mathcal{K}$-RelNet program, and a relative error of $0.8$ provides quite a small absolute error bound for the more reliable graphs in question.

\begin{figure*}[!t]
\centering
\subfloat[An example of method~(M\ref{item:rand})]{\includegraphics[width=2.5in]{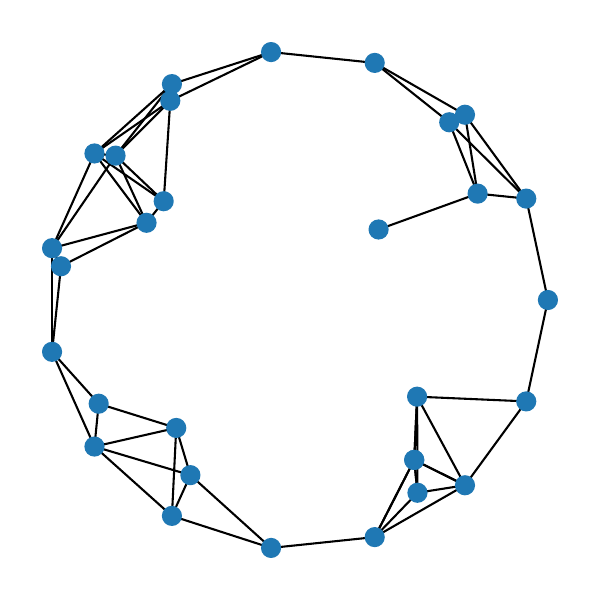}%
\label{subfig:item1}}
\hfil
\subfloat[An example of method~(M\ref{item:rand_spring})]{\includegraphics[width=2.5in]{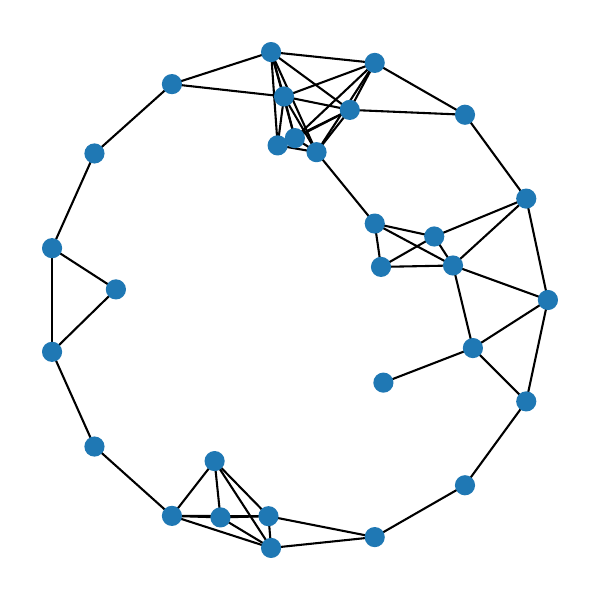}%
\label{subfig:item2}}
\\
\subfloat[An example of method~(M\ref{item:buffer})]{\includegraphics[width=2.5in]{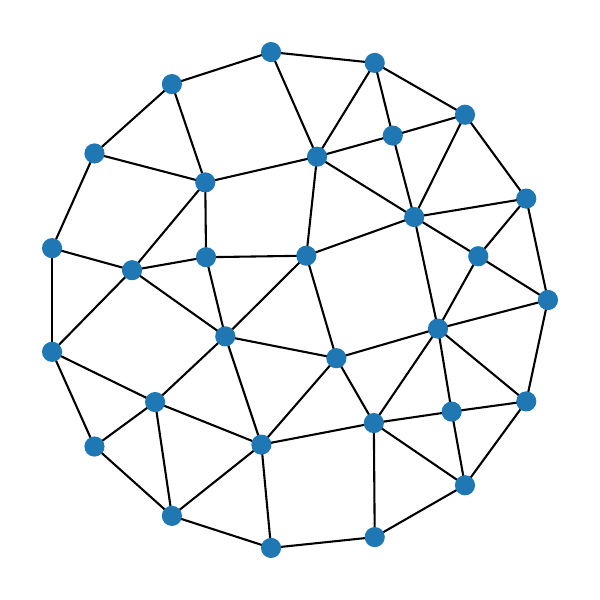}%
\label{subfig:item3}}
\hfil
\subfloat[An example of method~(M\ref{item:buffer_spring})]{\includegraphics[width=2.5in]{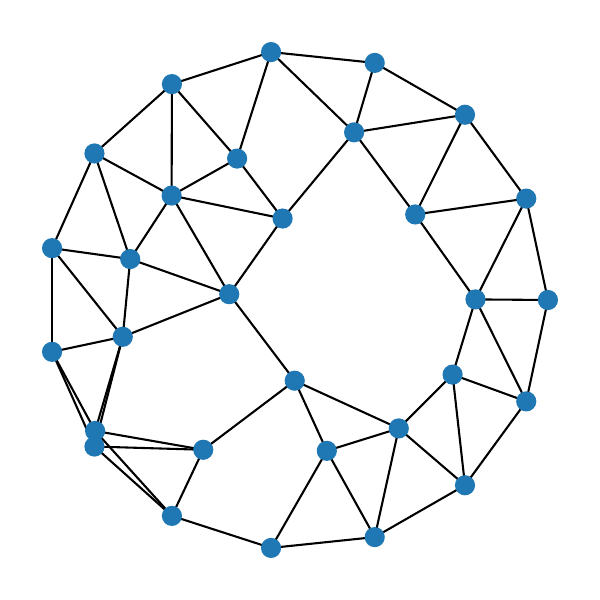}%
\label{subfig:item4}}
\caption{Fifteen vertices are added to the interior of a region outlined by a regular $15$-gon with edge-lengths of $0.9$ according to methods~(M\ref{item:rand})-(M\ref{item:buffer_spring}).}
\label{fig:sim_pics}
\end{figure*}


As a measure of area coverage, we calculate the radius of the largest empty circle (containing no vertices of the graph) whose center is contained in the convex hull of the vertex set.
The region will be well-covered if the radius of the largest empty circle is small.
A largest empty circle whose center is contained in the convex hull of a set of $n$ points can be found in $\mathcal{O}(n \log{n})$ time~\cite{toussaint1983computing}.
In Table~\ref{table:simulations}, we record the average all-terminal reliability and the average radius of the largest empty circle over $100$ graphs obtained using each of the methods~(M\ref{item:rand})-(M\ref{item:buffer_spring}).
The results in Table~\ref{table:simulations} are compared, and standard deviations depicted, in Figure~\ref{fig:comparisons}.

\begin{table}
    \caption{The average all-terminal reliability $\Rel_A$ with edge-operation probabilities of $90\%$ and the average radius of the largest empty circle (LEC) over $100$ graphs obtained using each method~(M\ref{item:rand})-(M\ref{item:buffer_spring}).}
    \centering
    \begin{tabular}{c || c | c | c | c}
        & (M\ref{item:rand}) & (M\ref{item:rand_spring}) & (M\ref{item:buffer}) & (M\ref{item:buffer_spring}) \\ \hline \hline
        Avg.\:$\Rel_A(G)$ & 81.62\% & 85.16\% & 99.09\% & 99.20\% \\ \hline
        Avg.\:radius of LEC & 1.2445 & 1.4287 & 0.6728 & 0.9697 \\
    \end{tabular}
    \label{table:simulations}
\end{table}

\begin{figure*}[!t]
        \centering
        \subfloat[(M\ref{item:rand}) versus (M\ref{item:rand_spring})]{\includegraphics[width=2.5in]{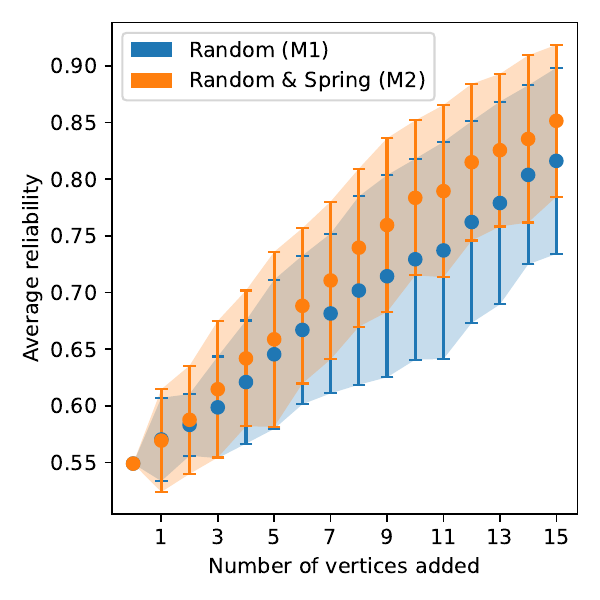}%
        \label{fig:plot_randVrandspring}}
        \hfil
        \subfloat[(M\ref{item:buffer}) versus (M\ref{item:rand_spring})]{\includegraphics[width=2.5in]{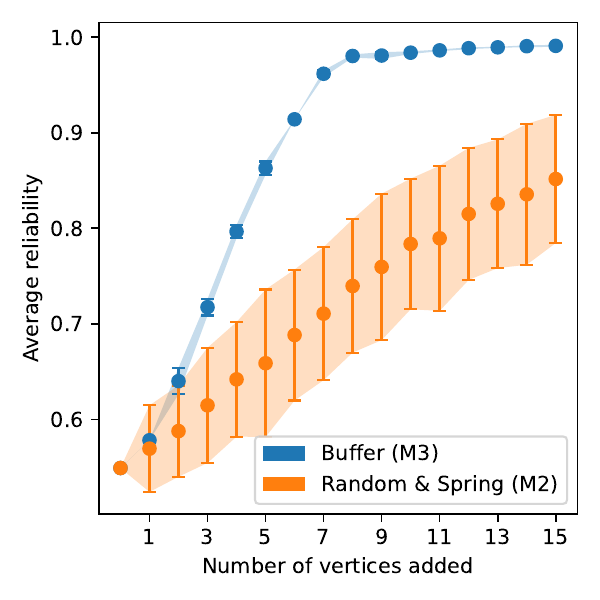}%
        \label{fig:plot_buffernospringVrandspring}}
        \\
        \subfloat[(M\ref{item:buffer_spring}) versus (M\ref{item:rand_spring})]{\includegraphics[width=2.5in]{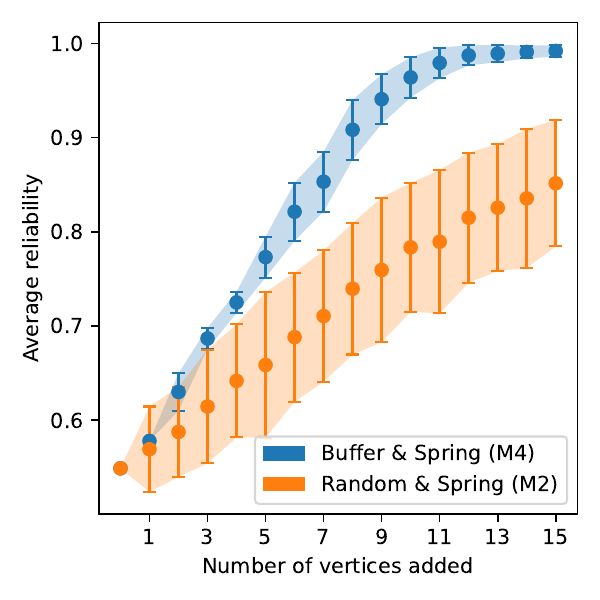}%
        \label{fig:plot_bufferspringVrandspring}}
        \caption{A comparison of average all-terminal reliability (with edge-operation probabilities of $90\%$) of graphs obtained via methods (M\ref{item:rand})-(M\ref{item:buffer_spring}). Shaded regions depict one standard deviation.}
        \label{fig:comparisons}
\end{figure*}

We first consider methods~(M\ref{item:rand}) and (M\ref{item:rand_spring}).
Averaged over $100$ trials, the all-terminal reliability of a graph obtained by adding $15$ vertices one at a time uniformly at random to the interior of the region, conditioned on connectedness of the resulting graph, is about $81.62\%$.
The average radius of the largest empty circle is about $1.245$.

The Fruchterman-Reingold algorithm has quite a different effect when applied to the graphs obtained via method (M\ref{item:rand}) than one might expect: vertices are pushed away from the interior of the region towards the boundary.
This increases, rather than decreases, the uncovered area of the region.
On the other hand, it also has the effect of slightly increasing the reliability.
We compared both types of our modified Fruchterman-Reingold algorithm from Section~\ref{sec:spring} (either all non-fixed vertices move at each iteration, or vertices move one at a time), and found that the former method was more effective for both reliability and area coverage purposes.

Averaged over $100$ trials, the all-terminal reliability of the graph obtained by adding $15$ vertices one at a time uniformly at random to the interior of the region, applying our modified Fruchterman-Reingold algorithm after each added vertex, is about $85.16\%$.
The average radius of the largest empty circle is about $1.43$.
Figure~\ref{fig:plot_randVrandspring} plots the average reliabilities obtained using method (M\ref{item:rand}) against those using method (M\ref{item:rand_spring}).

Using methods~(M\ref{item:rand}) and~(M\ref{item:rand_spring}) as a baseline, we now consider method~(M\ref{item:buffer}).
Having experimented with various buffer values $b$, we noted that, for $b \geq 0.75$, the region outlined by the fifteen satellites in Figure~\ref{fig:15-gon} tends to fill up before fifteen vertices can be added to the interior, meaning that there are no more points within the outlined region which are of distance at least $b$ from every vertex.
Note that, if the region fills up like this, then the radius of the largest empty circle is at most $b$.
We chose the value $b = 0.65$ as it fills the region relatively evenly and nearly fully after fifteen iterations.
Of course, depending on the imaging device used by the agents in the application, this value of $b$ may need to be altered accordingly.

We average the all-terminal reliability of $100$ graphs obtained by adding $15$ vertices to the interior of a region outlined by $15$ satellites using Algorithm~\ref{alg:enumerate_buffer_regions} ($b = 0.65$) and choosing the most reliable candidate at each step.
The average reliability after adding $15$ vertices is about $99.09\%$.
The average radius of the largest empty circle is about $0.67$.
The results are displayed in Figure~\ref{fig:plot_buffernospringVrandspring}.

Finally, we consider method~(M\ref{item:buffer_spring}).
Each time we add a vertex to the interior of the region to maximize reliability while being of distance at least $0.65$ from every other vertex, we apply our Fruchterman-Reingold algorithm from Section~\ref{sec:spring}.
We again use the version of our algorithm in which all non-fixed vertices move at each iteration.
Note that this algorithm may move vertices which were placed outside of buffer distance closer together.
Indeed, the largest empty circle is, on average, larger in this case.
Over $100$ trials, the average radius of the largest empty circle is about 
$0.97$.
The average all-terminal reliability after adding $15$ vertices is about
$99.2\%$.
The results are displayed in Figure~\ref{fig:plot_bufferspringVrandspring}.

\section{Future directions}\label{sec:future_directions}

In this paper, we considered various possible solutions to the problem of producing unit ball graphs with high reliability whose vertices are evenly spread over a given region.
We found that placing vertices one at a time in locations which are at least some fixed distance $b$ from all other vertices (using Algorithm~\ref{alg:enumerate_buffer_regions}), and which maximize the reliability over all such locations (using Monte Carlo simulations) provides a better method than using a modified spring layout algorithm for unit ball graphs.

These methods focused on balancing the reliability and area coverage objective functions.
Using the spring layout method, we attempted to spread out vertices while retaining the reliability of the original graph.
It would be of interest to accomplish this task using a less heuristic method.
We thus pose the following.

\begin{problem}
    Design an algorithm which takes as input a unit ball graph $G$ and outputs a new unit ball graph, containing all edges of $G$, whose largest empty circle centered in the convex hull of its vertex set is as small as possible.
\end{problem}

We now turn to the method which uses Algorithm~\ref{alg:enumerate_buffer_regions}: adding vertices one at a time to maximize reliability while enforcing a buffer distance.
While the estimated reliabilities of the graphs obtained using methods~(M\ref{item:buffer}) and~(M\ref{item:buffer_spring}) were quite high, even more reliable graphs may have been obtained by adding more than one vertex at a time.
Algorithms which approximate the possible neighborhoods for a set of vertices added in a given configuration are not hard to develop, but to list all such neighborhoods will require more complicated techniques.

\begin{problem}\label{prob:2ormore}
    Design an algorithm to find the locations to add two or more vertices to a unit ball graph and maximize reliability.
\end{problem}

In a similar vein, we have not fully resolved the problem of finding the most reliable unit ball graphs of a given order which cover a given area (even for the circular area considered in Section~\ref{sec:formation_planning}).
This problem is likely a difficult one, as finding most reliable (abstract) graphs of a given order and size remains an open problem.
However, the constraints here are quite different, and we pose this problem for completeness.

\begin{problem}
    Let $G$ be a unit disk graph whose vertices outline the boundary of a convex region, and let $p \in (0,1)$.
    Given a positive integer $k$ and positive real number $\ell$, find a unit disk graph obtained by adding $k$ vertices to $G$ whose largest empty circle has radius at most $\ell$ (should such a graph exist) which maximizes $\Rel_A$ (edge-operation probabilities $p$) over all such graphs.
\end{problem}

{\appendices
\section{Vertices on the boundary}\label{appendix:vxs_on_boundary} 

In our proof of Theorem~\ref{thm:buffer}, we claimed that, in order to find the set of all regions of intersection of a set of annuli centered at the points $v$ in $V$, it suffices to find the subsets of $V$ contained in a $(b,1)$-annulus with no points from $V$ in its inner circle, and with two points from $V$ on its boundary.
This claim follows from the following proposition.

\begin{proposition}\label{prop:vxs_on_boundary}
    Let $V \subset \mathbb{R}^2$ be a finite set of points in general position, and let $\emptyset \neq S \subseteq V$.
    If $S$ is contained in a $(b,c)$-annulus with all points in $V - S$ outside of its outer boundary, then there exists a second $(b,c)$-annulus with two points $u,v \in V$ on its boundary which contains $S - \{u,v\}$ and is such that all points in $V - (S \cup \{u,v\})$ lie outside of the outer boundary.
    
    Further, if $u$ and $v$ lie on the outer boundary of a $(b,c)$-annulus which contains $S$ and no points within its inner circle, then there are $(b,c)$-annuli containing each of $S$, $S \cup \{u\}$, $S \cup \{v\}$, and $S \cup \{u,v\}$, with no points on the boundary, and the rest of the points in $V$ outside of the outer circle.
    If one of $u$ or $v$ lies on the inner boundary, say $u$, then only two of these sets can be contained in such an annulus: $S \cup \{u\}$ and $S \cup \{u,v\}$.
    If both $u$ and $v$ lie on the inner boundary, then only $S \cup \{u,v\}$ is possible.
\end{proposition}
\begin{proof}
    First, let $A$ be a $(b,c)$-annulus with $S \subset A$ and $V - S$ lying outside of the outer boundary of $A$.
    It is not hard to see that we can shift $A$ slightly so that some point $u \in V$ is on its outer or inner boundary. We now rotate this shifted annulus around the fixed point $u$ until a second point $v$ lies on its boundary to obtain the desired annulus.

    On the other hand, suppose that $A'$ is a $(b,c)$-annulus with $u,v \in V$ on its boundary and no points in $V$ within its inner circle.
    Let $S$ denote the set of points contained in $A'$ (note $u,v \notin S$).
    Since the points in $V$ are in general position, we may assume that $u$ and $v$ do not lie on a diameter of $A'$.
    First, suppose $u$ lies on the outer boundary. 
    We can rotate $A'$ around $u$ to obtain an annulus containing $S \cup \{v\}$ with only $u$ on its boundary.
    Then, shifting the resulting annulus slightly, we can obtain annuli containing either $S \cup \{u,v\}$ or $S \cup \{v\}$ with all other points in $V$ outside of the its outer circle.
    Similarly, if $v$ is on the outer boundary of $A'$, we could have rotated $A'$ in the other direction to obtain an annulus containing $S$ but with $V - (S \cup \{u\})$ outside of its outer boundary.
    In this way, we obtain annuli containing all four sets $S$, $S \cup \{u\}$, $S \cup \{v\}$, and $S \cup \{u,v\}$ in their interiors, and no other points in their inner circle.

    On the other hand, if $u$ is on the inner boundary, then any shift of the annulus puts $u$ inside the inner circle if not in the annulus.
    Thus, only the sets $S \cup \{u\}$ and $S \cup \{u,v\}$ are possibilities in case $v$ is on the outer boundary, and only $S \cup \{u,v\}$ is possible if both are on the inner boundary of $A'$.
\end{proof}

\section{InitialSet algorithm}\label{appendix:InitialSet}

\begin{algorithm}[H]
    \caption{Determine initial set of vertices in annulus when a sweep begins}
    \label{alg:initial_set}
    \begin{algorithmic}[1]
        \Require{$V \subset \mathbb{R}^2$ and output $L$ of Algorithm~\ref{alg:outer_sweep} or Algorithm~\ref{alg:inner_sweep}}
        \Ensure{Subset of $V$ contained in annulus when sweep begins}
        \Statex

        \Function{InitialSet}{$L$}
        \Let{$R$}{$\{ \}$}
        \For{$(w, \tmin, \tdis, \tre, \tmax) \in L$}
            \If{$\tre$ is None and $\tmax < \tmin$}
                \State{add $(w, \tmin, \mathrm{type}=\min)$ to $R$}
            \ElsIf{$\tdis < \tre$ and $\tmax$ is None}
                \State{add $(w, \tre, \mathrm{type}=\mathrm{re})$ to $R$}
            \ElsIf{$\tdis < \tre$ and $\tmax < \tmin$}
                \If{$\tdis < 0$}
                    \State{add $(w, \tmin, \mathrm{type}=\min)$ to $R$}
                \ElsIf{$\tre > 0$}
                    \State{add $(w, \tre, \mathrm{type}=\mathrm{re})$ to $R$}
                \EndIf
            \EndIf
        \EndFor
        \State \Return{$R$}
        \EndFunction
    \end{algorithmic}
\end{algorithm}
}


\begin{thebibliography}{10}
\providecommand{\url}[1]{#1}
\csname url@samestyle\endcsname
\providecommand{\newblock}{\relax}
\providecommand{\bibinfo}[2]{#2}
\providecommand{\BIBentrySTDinterwordspacing}{\spaceskip=0pt\relax}
\providecommand{\BIBentryALTinterwordstretchfactor}{4}
\providecommand{\BIBentryALTinterwordspacing}{\spaceskip=\fontdimen2\font plus
\BIBentryALTinterwordstretchfactor\fontdimen3\font minus
  \fontdimen4\font\relax}
\providecommand{\BIBforeignlanguage}[2]{{%
\expandafter\ifx\csname l@#1\endcsname\relax
\typeout{** WARNING: IEEEtran.bst: No hyphenation pattern has been}%
\typeout{** loaded for the language `#1'. Using the pattern for}%
\typeout{** the default language instead.}%
\else
\language=\csname l@#1\endcsname
\fi
#2}}
\providecommand{\BIBdecl}{\relax}
\BIBdecl

\bibitem{morgan2012swarm}
D.~Morgan, S.-J. Chung, L.~Blackmore, B.~Acikmese, D.~Bayard, and F.~Y.
  Hadaegh, ``Swarm-keeping strategies for spacecraft under {$J_2$} and
  atmospheric drag perturbations,'' \emph{Journal of Guidance, Control, and
  Dynamics}, vol.~35, no.~5, pp. 1492--1506, 2012.

\bibitem{stoica}
\BIBentryALTinterwordspacing
{JPL Robotics}. ``Swarm autonomy.'' www-robotics.jpl.nasa.gov. Accessed: Jan. 11, 2024. [Online].
  Available:
  \url{https://www-robotics.jpl.nasa.gov/what-we-do/research-tasks/swarm-autonomy/}
\BIBentrySTDinterwordspacing

\bibitem{taxonomy}
\BIBentryALTinterwordspacing
{NASA}. ``2020 {NASA} technology taxonomy.'' nasa.gov. Accessed: Jan. 11, 2024. [Online].
  Available: \url{https://www.nasa.gov/otps/2020-nasa-technology-taxonomy/}
\BIBentrySTDinterwordspacing

\bibitem{fishman}
\BIBentryALTinterwordspacing
J.~Fishman. ``{NASA} small satellites to demonstrate swarm communications and
  autonomy.'' nasa.gov. Accessed: Jan. 11, 2024. [Online]. Available:
  \url{https://www.nasa.gov/smallspacecraft/nasa-small-satellites-to-demonstrate-swarm-communications-and-autonomy/}
\BIBentrySTDinterwordspacing

\bibitem{9153840}
J.~Connor, B.~Champion, and M.~A. Joordens, ``Current algorithms, communication
  methods and designs for underwater swarm robotics: a review,'' \emph{IEEE
  Sensors Journal}, vol.~21, no.~1, pp. 153--169, 2021.

\bibitem{budiharto2019review}
W.~Budiharto, A.~Chowanda, A.~A.~S. Gunawan, E.~Irwansyah, and J.~S. Suroso,
  ``A review and progress of research on autonomous drone in agriculture,
  delivering items and geographical information systems ({GIS}),'' in
  \emph{2019 2nd World Symposium on Communication Engineering (WSCE)}.\hskip
  1em plus 0.5em minus 0.4em\relax IEEE, 2019, pp. 205--209.

\bibitem{saffre2023wild}
F.~Saffre, H.~Karvonen, and H.~Hildmann, ``Wild swarms: Autonomous drones for
  environmental monitoring and protection,'' in \emph{International conference
  on FinDrones}.\hskip 1em plus 0.5em minus 0.4em\relax Springer, 2023, pp.
  1--32.

\bibitem{hu2020voronoi}
J.~Hu, H.~Niu, J.~Carrasco, B.~Lennox, and F.~Arvin, ``Voronoi-based
  multi-robot autonomous exploration in unknown environments via deep
  reinforcement learning,'' \emph{IEEE Transactions on Vehicular Technology},
  vol.~69, no.~12, pp. 14\,413--14\,423, 2020.

\bibitem{dorigo2021swarm}
M.~Dorigo, G.~Theraulaz, and V.~Trianni, ``Swarm robotics: Past, present, and
  future [point of view],'' \emph{Proceedings of the IEEE}, vol. 109, no.~7,
  pp. 1152--1165, 2021.

\bibitem{farrag2021satellite}
A.~Farrag, S.~Othman, T.~Mahmoud, and A.~Y. ELRaffiei, ``Satellite swarm survey
  and new conceptual design for earth observation applications,'' \emph{The
  Egyptian Journal of Remote Sensing and Space Science}, vol.~24, no.~1, pp.
  47--54, 2021.

\bibitem{chen2024reliability}
Z.~Chen, H.~Zhang, X.~Wang, J.~Yang, and H.~Dui, ``Reliability analysis and
  redundancy design of satellite communication system based on a novel
  {B}ayesian environmental importance,'' \emph{Reliability Engineering \&
  System Safety}, vol. 243, p. 109813, 2024.

\bibitem{ESA}
\BIBentryALTinterwordspacing
{The European Space Agency}. ``About space debris.'' esa.int. Accessed: Jan. 16, 2024.
  [Online]. Available:
  \url{https://www.esa.int/Space_Safety/Space_Debris/About_space_debris}
\BIBentrySTDinterwordspacing

\bibitem{Howell}
\BIBentryALTinterwordspacing
E.~Howell. ``Van {A}llen radiation belts: Facts \& findings.'' space.com.
  Accessed: Jan. 15, 2024. [Online]. Available:
  \url{https://www.space.com/33948-van-allen-radiation-belts.html}
\BIBentrySTDinterwordspacing

\bibitem{SoA}
\BIBentryALTinterwordspacing
{NASA}. ``State-of-the-art of small spacecraft technology.'' nasa.gov.
  Accessed: Jan. 15, 2024. [Online]. Available:
  \url{https://www.nasa.gov/smallsat-institute/sst-soa/soa-communications/}
\BIBentrySTDinterwordspacing

\bibitem{aboelfotoh_computing_1989}
H.~AboElFotoh and C.~Colbourn, ``Computing 2-terminal reliability for
  radio-broadcast networks,'' \emph{IEEE Transactions on Reliability}, vol.~38,
  no.~5, pp. 538--555, Dec. 1989.

\bibitem{colbourn_combinatorics_1987}
C.~J. Colbourn, \emph{The combinatorics of network reliability}.\hskip 1em plus
  0.5em minus 0.4em\relax Oxford University Press, Inc., 1987.

\bibitem{shpungin_combinatorial_2006}
Y.~Shpungin, ``Combinatorial approach to reliability evaluation of network with
  unreliable nodes and unreliable edges,'' \emph{International Journal of
  Computer Science}, vol.~1, no.~3, pp. 177--183, 2006.

\bibitem{dagum2000optimal}
P.~Dagum, R.~Karp, M.~Luby, and S.~Ross, ``An optimal algorithm for {M}onte
  {C}arlo estimation,'' \emph{SIAM Journal on computing}, vol.~29, no.~5, pp.
  1484--1496, 2000.

\bibitem{ball_chapter_1995}
M.~O. Ball, C.~J. Colbourn, and J.~S. Provan, ``\BIBforeignlanguage{en}{Chapter
  11 {Network} reliability},'' in \emph{\BIBforeignlanguage{en}{Handbooks in
  {Operations} {Research} and {Management} {Science}}}, ser. Network
  {Models}.\hskip 1em plus 0.5em minus 0.4em\relax Elsevier, Jan. 1995, vol.~7,
  pp. 673--762.

\bibitem{buchanan2023node}
C.~Buchanan, J.~Bagrow, P.~Rombach, and H.~Ossareh, ``Node placement to
  maximize reliability of a communication network with application to satellite
  swarms,'' in \emph{2023 IEEE International Conference on Systems, Man, and
  Cybernetics (SMC)}, 2023, pp. 3466--3473.

\bibitem{fruchterman1991graph}
T.~M. Fruchterman and E.~M. Reingold, ``Graph drawing by force-directed
  placement,'' \emph{Software: Practice and experience}, vol.~21, no.~11, pp.
  1129--1164, 1991.

\bibitem{jindal2023}
\BIBentryALTinterwordspacing
A.~Jindal. ``Angular sweep (maximum points that can be enclosed in a circle of
  given radius).'' GeeksforGeeks.org. [Online]. Available:
  \url{https://www.geeksforgeeks.org/angular-sweep-maximum-points-can-enclosed-circle-given-radius/}
\BIBentrySTDinterwordspacing

\bibitem{yaglom1964challenging}
A.~M. Yaglom and I.~M. Yaglom, \emph{Challenging mathematical problems with
  elementary solutions}.\hskip 1em plus 0.5em minus 0.4em\relax Holden-Day, San
  Francisco, 1964, vol. I, Combinatorial Analysis and Probability Theory,
  translated by J. McCawley, Jr., revised and edited by B. Gordon.

\bibitem{breu1998unit}
H.~Breu and D.~G. Kirkpatrick, ``Unit disk graph recognition is {NP}-hard,''
  \emph{Computational Geometry}, vol.~9, no. 1-2, pp. 3--24, 1998.

\bibitem{megiddo1983linear}
N.~Megiddo, ``Linear-time algorithms for linear programming in {$R^3$} and
  related problems,'' \emph{SIAM Journal on Computing}, vol.~12, no.~4, pp.
  759--776, 1983.

\bibitem{welzl2005smallest}
E.~Welzl, ``Smallest enclosing disks (balls and ellipsoids),'' in \emph{New
  Results and New Trends in Computer Science}, ser. Lecture Notes in Computer
  Science, H.~Maurer, Ed., vol. 555.\hskip 1em plus 0.5em minus 0.4em\relax
  Springer, Berlin, Heidelberg, 1991, pp. 359--370.

\bibitem{krupke_distributed_2015-1}
D.~Krupke, M.~Ernestus, M.~Hemmer, and S.~P. Fekete, ``Distributed cohesive
  control for robot swarms: {Maintaining} good connectivity in the presence of
  exterior forces,'' in \emph{2015 {IEEE}/{RSJ} {International} {Conference} on
  {Intelligent} {Robots} and {Systems} ({IROS})}, Sep. 2015, pp. 413--420.

\bibitem{deng2019optimized}
X.~Deng, Z.~Yu, R.~Tang, X.~Qian, K.~Yuan, and S.~Liu, ``An optimized node
  deployment solution based on a virtual spring force algorithm for wireless
  sensor network applications,'' \emph{Sensors}, vol.~19, no.~8, p. 1817, 2019.

\bibitem{li2022fruchterman}
J.~Li, Y.~Tao, K.~Yuan, R.~Tang, Z.~Hu, W.~Yan, and S.~Liu,
  ``Fruchterman--{R}eingold hexagon empowered node deployment in wireless
  sensor network application,'' \emph{Sensors}, vol.~22, no.~14, p. 5179, 2022.

\bibitem{howard2002mobile}
A.~Howard, M.~J. Matari{\'c}, and G.~S. Sukhatme, ``Mobile sensor network
  deployment using potential fields: A distributed, scalable solution to the
  area coverage problem,'' in \emph{Distributed autonomous robotic systems
  5}.\hskip 1em plus 0.5em minus 0.4em\relax Springer, 2002, pp. 299--308.

\bibitem{qi2022wireless}
X.~Qi, Z.~Li, C.~Chen, and L.~Liu, ``A wireless sensor node deployment scheme
  based on embedded virtual force resampling particle swarm optimization
  algorithm,'' \emph{Applied Intelligence}, vol.~52, no.~7, pp. 7420--7441,
  2022.

\bibitem{SciPyProceedings_11}
A.~A. Hagberg, D.~A. Schult, and P.~J. Swart, ``Exploring network structure,
  dynamics, and function using {N}etwork{X},'' in \emph{Proceedings of the 7th
  Python in Science Conference}, G.~Varoquaux, T.~Vaught, and J.~Millman, Eds.,
  Pasadena, CA USA, 2008, pp. 11 -- 15.

\bibitem{karp_monte-carlo_1985}
R.~M. Karp and M.~Luby, ``\BIBforeignlanguage{en}{Monte-{Carlo} algorithms for
  the planar multiterminal network reliability problem},''
  \emph{\BIBforeignlanguage{en}{Journal of Complexity}}, vol.~1, no.~1, pp.
  45--64, Oct. 1985.

\bibitem{paredes2019principled}
R.~Paredes, L.~Due{\~n}as-Osorio, K.~S. Meel, and M.~Y. Vardi, ``Principled
  network reliability approximation: A counting-based approach,''
  \emph{Reliability Engineering \& System Safety}, vol. 191, p. 106472, 2019.

\bibitem{chakraborty2013scalable}
S.~Chakraborty, K.~S. Meel, and M.~Y. Vardi, ``A scalable approximate model
  counter,'' in \emph{Principles and Practice of Constraint Programming: 19th
  International Conference, CP 2013, Uppsala, Sweden, September 16-20, 2013.
  Proceedings 19}.\hskip 1em plus 0.5em minus 0.4em\relax Springer, 2013, pp.
  200--216.

\bibitem{yang2023rounding}
J.~Yang and K.~S. Meel, ``Rounding meets approximate model counting,'' in
  \emph{International Conference on Computer Aided Verification}.\hskip 1em
  plus 0.5em minus 0.4em\relax Springer, 2023, pp. 132--162.

\bibitem{toussaint1983computing}
G.~T. Toussaint, ``Computing largest empty circles with location constraints,''
  \emph{International Journal of Computer \& Information Sciences}, vol.~12,
  pp. 347--358, 1983.

\end{thebibliography}


\end{document}